\documentclass[12pt,onecolumn,draftclsnofoot]{IEEEtran}
\usepackage{cite}      % Written by Donald Arseneau
\usepackage{graphicx}  % Written by David Carlisle and Sebastian Rahtz

\usepackage{amsmath}   % From the American Mathematical Society
\newtheorem{theorem}{Theorem}
\newtheorem{lemma}{Lemma}

\def\E{{\rm E}\,}

\DeclareMathOperator*{\argmax}{arg\,max} 

\newcommand{\bx}{\ensuremath{\boldsymbol{x}}}
\newcommand{\by}{\ensuremath{\boldsymbol{y}}}

\newcommand{\blambda}{\ensuremath{\boldsymbol{\lambda}}}

\newcommand{\tblambda}{\ensuremath{\boldsymbol{\tilde{\lambda}}}}
\newcommand{\tlambda}{\ensuremath{\tilde{\lambda}}}

\newcommand{\bgamma}{\ensuremath{\boldsymbol{\gamma}}}
\newcommand{\tbgamma}{\ensuremath{\boldsymbol{\tilde{\gamma}}}}
\newcommand{\tgamma}{\ensuremath{\tilde{\gamma}}}

\newcommand{\ba}{\ensuremath{\boldsymbol{a}}}
\newcommand{\bb}{\ensuremath{\boldsymbol{b}}}

\newcommand{\bxi}{\ensuremath{\boldsymbol{x_i}}}
\newcommand{\bVi}{\ensuremath{\boldsymbol{v_i}}}
\newcommand{\bxj}{\ensuremath{\boldsymbol{x_j}}}
\newcommand{\bVj}{\ensuremath{\boldsymbol{v_j}}}
\newcommand{\bV}{\ensuremath{\boldsymbol{v}}}
\newcommand{\bF}{\ensuremath{\boldsymbol{f}}}
\newcommand{\bFi}{\ensuremath{\boldsymbol{f_i}}}
\newcommand{\bFj}{\ensuremath{\boldsymbol{f_j}}}
\newcommand{\bxin}{\ensuremath{\boldsymbol{X_i}}}
\newcommand{\bxjn}{\ensuremath{\boldsymbol{X_j}}}
\newcommand{\byn}{\ensuremath{\boldsymbol{Y}}}
\newcommand{\bxn}{\ensuremath{\boldsymbol{X}}}

\newcommand{\bVcap}{\ensuremath{\boldsymbol{V}}}

\newcommand{\cac}{\ensuremath{\mathcal{C}}}
\newcommand{\cd}{\ensuremath{\mathcal{D}}}

\newcommand{\cp}{\ensuremath{\mathcal{P}}}

\newcommand{\cs}{\ensuremath{\mathcal{S}}}
\newcommand{\ct}{\ensuremath{\mathcal{T}}}

\newcommand{\cv}{\ensuremath{\mathcal{V}}}

\newcommand{\cx}{\ensuremath{\mathcal{X}}}
\newcommand{\cy}{\ensuremath{\mathcal{Y}}}

\begin{document}
\title{The Sensing Capacity of Sensor Networks}
\author{Yaron~Rachlin,~\IEEEmembership{Student~Member,~IEEE,}
        Rohit~Negi,~\IEEEmembership{Member,~IEEE,}
        and~Pradeep~Khosla,~\IEEEmembership{Fellow,~IEEE} \thanks{All authors affiliated with the Department of Electrical and Computer Engineering at Carnegie Mellon University.}}% <-this % stops a space
\maketitle

\begin{abstract}
This paper demonstrates fundamental limits of sensor networks for detection problems
where the number of hypotheses is exponentially large. Such problems characterize many
important applications including detection and classification of targets in a
geographical area using a network of sensors, and detecting complex substances with a
chemical sensor array. We refer to such applications as large-scale detection problems.
Using the insight that these problems share fundamental similarities with the problem of
communicating over a noisy channel, we define a quantity called the `sensing capacity'
and lower bound it for a number of sensor network models. The sensing capacity expression
differs significantly from the channel capacity due to the fact that a fixed sensor
configuration encodes all states of the environment. As a result, codewords are dependent
and non-identically distributed. The sensing capacity provides a bound on the minimal
number of sensors required to detect the state of an environment to within a desired
accuracy. The results differ significantly from classical detection theory, and provide
an intriguing connection between sensor networks and communications. In addition, we
discuss the insight that sensing capacity provides for the problem of sensor selection.
\end{abstract}


% Will be ignored for other modes.
\IEEEpeerreviewmaketitle

\section{Introduction}
\label{sec.intro}

\PARstart{A}\  sensor network is deployed to obtain information about the state of an
environment using multiple sensors. In many sensing applications, such as pollution
monitoring and border security, the phenomena under observation has a large scale that
exceeds the range of any one sensor. As a result, collecting measurements from multiple
sensors is essential to the sensing task. Obtaining information about an environment can
be cast as either a `detection' or an `estimation' problem. In estimation problems such
as the problem of estimating a continuous field to within a desired distortion, the state
of the environment is continuous. In detection problems, such as binary hypothesis
testing, the state of the environment is represented as a finite set of hypotheses. In
this paper we study the problem of `large-scale detection' where the state of the
environment belongs to an exponentially large, structured set of hypotheses. Large-scale
detection problems include many applications where a sensor network is deployed in order
to monitor a large-scale phenomena. We exploit the structure of large-scale detection
problems to demonstrate a fundamental information-theoretic relationship between the
number of sensor measurements and ability of a sensor network to detect the state of the
environment to within a desired accuracy.

We obtain our results by drawing on an analogy between sensor networks and channel
encoders. For a fixed sensor configuration, each state of the environment induces a
corresponding set of sensor outputs. This set of sensor outputs can be viewed as a
noise-corrupted `codeword,' which must be `decoded' in order to detect the state of the
environment. Thus, the sensor network acts as a channel encoder. In order to motivate
this perspective, we examine the following large-scale detection applications.

%\subsection{Motivating Sensor Network Applications}
%\label{sec.intro.motivating}

%------------------------------

%robotic mapping
Robotic mapping is the first large-scale application we consider \cite{Thrun02a}. In
mapping, robots collect sensor measurements to map an unknown environment for the purpose
of navigation. \cite{Elfes89} introduced occupancy grids, one of the most popular
approaches to this problem. In occupancy grids, the world is modeled as a discrete grid,
where each grid location has a value corresponding to the state of the environment. For
example, in a binary a grid a `0' can indicate free space while a `1' can indicate an
obstacle. A robot traversing an unknown environment collects sensor measurements that
encode the state of the environment. For example, a robot using a sonar sensor emits a
wide acoustic pulse and measures the time until a reflected pulse is sensed. These
readings are ambiguous, since one cannot infer the precise location of the obstacle that
caused the reflection from a single sensor reading. In addition, sonar readings are
noisy. As a result, multiple sensor measurements must be used to distinguish among an
exponentially large number of possible grid states. The sequence of sonar readings can be
viewed as a noise-corrupted codeword corresponding to the state of the grid. While
robotic mapping systems have been successfully implemented in practice, little can be
said about their theoretical performance. Theoretical understanding could shed light on
the number of sensor measurements required to map an unknown environment. In addition,
theory can provide insight into questions about sensor selection. Is it better to use
cheap, low power, wide angle sensors or expensive, high power, narrow angle sensors? A
theoretical framework could provide general insight into such sensor selection questions.

%---------------------------

%video surveillance
Video surveillance is another large scale detection problem. \cite{Collins_2000_3325}
used multi-camera sensor networks to detect and track objects across multiple areas, and
\cite{Hoover99} uses multiple cameras to localize moving objects in a room. The region
under surveillance can be viewed as a three-dimensional grid. For example each grid
position can have a binary value, representing motion or lack of motion in that grid
position. As in the previous example, the number of states of this grid is exponential in
the number of grid blocks. Each camera observes a subset of grid blocks, but introduces
ambiguity by reducing a three-dimensional volume to a two-dimensional image. As a result
multiple camera images must be combined to detect the state of the environment. The set
of images encode the grid state. While practical systems for surveillance applications
are deployed, a theoretical framework for understanding performance limits for such
problems is not available.

%---------------------------

%chemical sensors - other citations
Identifying a complex chemical substance is a third example of a large-scale detection
problem. In this application the output of chemical sensor arrays, consisting of
heterogeneous chemical sensors, is used to distinguish among a large number of substances
\cite{Burl02}. Each substance can be modeled as a mixture of constituent chemicals at
various discrete concentration levels, resulting in an exponentially large set of
possible states. Each chemical sensor in the array reacts to a subset of chemicals. For
example, sensors can output a voltage proportional to a weighted sum of the
concentrations of a subset of chemicals. The output of a chemical sensor array encodes
the state of the sample being sensed. As in the previous two examples, theory could
provide insight into the practical design of such sensor arrays.

%---------------------------

%seismic
Target detection and classification in a geographical area is an important class of
applications for sensor networks \cite{Li02}, and a final motivating example of a
large-scale detection problem. We consider the problem of detection and classification
using seismic sensors, as demonstrated in \cite{Li02},\cite{Tian02}. The environment can
be modeled as a discrete grid, where each position can contain targets of multiple types.
The number of target configurations is exponential in the number of grid blocks. Seismic
sensors are scattered randomly on this grid, and sense the vibrations of targets over
subsets of the grid. The intensity of vibration is dependent on the target's distance
from the sensor, and therefore a single sensor cannot distinguish between many targets
far away and a single target nearby. The set of seismic sensor outputs encode the
location and class of targets in the field.

%---------------------------

All of the examples considered above share the following common elements. The state can
be modeled as a discrete vector or grid, and the number of states is exponentially large.
Sensors output noise-corrupted functions of subsets of the vector or grid. These sensor
measurements must be fused in order to detect the state of the environment. In this paper
we analyze the fundamental limits of this process by using the insight that the problem
of large-scale detection and the problem of communicating over a noisy channel share
essential similarities.

\section{Sensor Networks and Communication Channels}
\label{sec.sensorsandchannels}

The examples described in Section \ref{sec.intro} motivate the sensor network model shown
in Figure \ref{fig.sensor_channel}. A discrete target vector $\bV$ represents the state
of the environment. In this paper, the term `state' and `target vector' are used
interchangeably. A fixed sensor configuration encodes the state as a vector of noiseless
sensor outputs that form the codeword $\bx$. The observed, noisy sensor measurements are
written as $\by$, a noise-corrupted version of $\bx$. Finally, a detection algorithm uses
$\by$ to compute a guess of the state of the environment $\hat{\bV}$.

The sensor model shown in Figure \ref{fig.sensor_channel} is similar to the classical
communication channel model shown in Figure \ref{fig.communication_channel}. The target
vector $\bV$ corresponds to the message $m$ being sent. The sensor network acts as a
channel encoder, producing the codeword $\bx$. Finally, a detection algorithm acts as a
channel decoder on the noise corrupted codeword $\by$. Shannon's celebrated Channel
capacity results provide limits for the communications channel \cite{Shannon48}.
Motivated by the similarity between the sensor network model and the communication
channel model, we defined and bounded the \emph{sensing capacity} in \cite{RachlinITW04}.
The sensing capacity plays a role in our sensor network model analogous to the role of
channel capacity in a communications channel. However, because the models differ in
significant ways, the notions of channel capacity and sensing capacity also differ.

The most important difference between the sensor network model and a communication
channel model is at the encoder. In communications, the content of the message and its
codeword representation can be decoupled. Further, the channel encoder can implement any
mapping between message and codeword. As a result, two highly similar messages can be
differentiated with arbitrarily high accuracy. In contrast, a sensor network encoder uses
the same sensor configuration to encode all states of the environment. Further, since
sensors react to some phenomena in the environment and are limited by physical
constraints, the codeword associated with a particular state of the environment is a
direct function of that state. Therefore the state and its codeword representation are
coupled. As a result, two highly similar states of the environment cannot be
distinguished with arbitrarily high accuracy. While similarities between the sensor
network model and the channel model motivate the application of insights about
communications from information theory, significant differences between the two models
require care in applying such insights in order to understand the impact of these
differences on the final theoretical results.

Section \ref{sec.intro_overview} provides an overview of the main results presented in
this paper, and reviews related work. Section \ref{sec.itw} presents sensing capacity
results for non-spatial (e.g. chemical) sensing applications, while Section
\ref{sec.ipsn} demonstrates sensing capacity results for a sensor network model that
accounts for spatial locality in sensor observations. Section \ref{sec.concs} concludes
the paper and discusses future work.

\section{Main Results and Related Work}
\label{sec.intro_overview}

We review the main theoretical results presented in this paper. In Section \ref{sec.itw}
we introduce a simple but useful sensor network model that can be used to model sensing
applications such as chemical sensing applications and computer network monitoring. For
this model, we define and bound the sensing capacity. The sensing capacity bound differs
significantly from the standard channel capacity results, and requires novel arguments to
account for the constrained encoding of a sensor network. This is an important
observation due to the use of mutual information as a sensor selection heuristic
\cite{DurrantWhyte94}. Our result shows that this is not the correct metric for
large-scale detection applications. Extensions are presented to account for non-binary
target vectors, target sparsity, and heterogeneous sensors. Plotting the sensing capacity
bound, we demonstrate interesting sensing tradeoffs. For example, perhaps
counter-intuitively, sensors of shorter range can achieve a desired detection accuracy
with fewer measurements than sensors of longer range. Finally, we also compare our
sensing capacity bound to simulated sensor network performance.

In Section \ref{sec.ipsn} we introduce a sensor network model that accounts for
contiguity in a sensor's field of view. Contiguity is an essential aspect of many classes
of sensors. For example, cameras observe localized regions and seismic sensors sense
vibrations from nearby targets. We demonstrate sensing capacity bounds that account for
such sensors by extending results about Markov types \cite{Csiszar98}, and use convex
optimization to compute these bounds. The first result in Section \ref{sec.ipsn} assumes
the state of the environment is modeled as a one-dimensional vector. In Section
\ref{sec.2dfields} we extend this result to the case where the state of the environment
is modeled as a two-dimensional grid. While a one-dimensional vector can model sensor
network applications such as border security and traffic monitoring, results about two
dimensions significantly increase the type of applications described by our models.

%\section{Related Work}
%\label{sec.relatedwork}

The performance of sensor networks is limited by both sensing resources and non-sensing
resources such as communications, computation, and power. One set of results has been
obtained by considering the limitations that communications requirements impose on a
sensor network. \cite{DuarteMelo03} extends the results in \cite{Gupta00} to account for
the different traffic models that arise in a sensor network. \cite{Barrenchea04} studies
network transport capacity for the case of regular sensor networks. \cite{Mitra04}
studies the impact of computational constraints and power on the communication efficiency
of sensor networks. \cite{Hu04} has considered the interaction between transmission rates
and power constraints. Another set of results has been obtained by extending results from
compression to sensor networks. Distributed source coding \cite{Slepian73},\cite{Wyner76}
provides limits on the compression of separately encoded correlated sources.
\cite{Pradhan02} applies these results to sensor networks. \cite{Xiong04} provides an
overview of this area of research.

In contrast to the work mentioned above, we focus directly on the limits of detecting the
state of the environment using noisy sensor observations. The notion of sensing capacity
characterizes the limits that sensing (e.g. sensor type, range, and noise) imposes on the
attainable accuracy of detection. We do not examine the compression of sensor
observations, or the resources required to communicate sensor observations to a point in
the network. Instead, we focus on the limits of detection accuracy assuming complete
availability of noisy sensor observations. Among existing work in information theory, the
problem we investigate in this paper is unlike a source coding problem, and is similar to
a channel coding problem. However, the sensor network model we investigate is
fundamentally different than a standard communications channel.

Our work is most closely related to work on detection and classification in sensor
networks. \cite{Varshney97} describes a large body of work on distributed detection where
the number of hypotheses is small. \cite{Chamberland03}, \cite{Chamberland04} extend this
work to consider a decentralized binary detection problem with noisy communication links
to obtain error exponents. \cite{DCosta04} analyzes the performance of various
classification schemes for classifying a Gaussian source. This is an m-ary problem where
the number of hypotheses is small. \cite{Kotecha05} analyzes the performance suboptimal
classification schemes for classifying multiple targets. While the number of hypotheses
is exponential in the number of targets, the large-scale detection problem of a large
number of targets is not considered. \cite{Chakrabarty01} considers the problem of sensor
placement for detecting the location of one or few targets in a grid. This problem is
most closely related to the large-scale detection problems addressed in this paper.
However, due to restrictions on the numbers of targets, the number of hypotheses remains
small in comparison to a large-scale detection problem. A coding-based approach was used
to bound the minimum number of sensors required for discrimination, and to propose
structured sensor configurations. However, sensors were noiseless, and of limited type,
and no notion of sensing capacity was considered. In contrast to existing to existing
work on detection and classification in sensor networks, we demonstrate fundamental
performance limits for large-scale detection problems.

The problem of estimating a continuous field using a sensor network is an active area of
research. \cite{Scaglione02} considers the relationship of transport capacity and the
rate distortion function of a continuous random processes. \cite{Nowak04} proves limits
on the estimation of an inhomogeneous random fields using sensor that collect noisy point
samples. Other work on the problem of estimating a continuous random field includes
\cite{Marco03},\cite{Ishwar03}, \cite{Bajwa05},\cite{Kumar04}. \cite{Gastpar05} considers
the estimation of continuous parameters of a set of underlying random processes through a
noisy communications channel. The results presented in this paper consider the detection
of a discrete state of an environment. We do not consider extensions to environments with
a continuous state.

\section{Sensing Capacity of the Arbitrary Connections Model}
\label{sec.itw}

In this section we define and analyze the sensing capacity of the arbitrary connections
model, a simple but useful model introduced in \cite{RachlinITW04}. We denote random
variables and functions by upper-case letters, and instantiations or constants by
lower-case letters. Bold-font denotes vectors. $\log(\cdot)$ has base-2. Sets are denoted
using calligraphic script. $D(P||Q)$ denotes the Kullback-Leibler distance and $H(P)$
denotes entropy of a random variable with probability distribution $P$. $H(Q|P)$ is the
conditional entropy of a random variable with conditional probability distribution $Q$
given another random variable with probability distribution $P$.

\subsection{Arbitrary Connections Model}
\label{sec.arbitraryconnsmodel}

Figure \ref{fig.ind_model} shows an example of the arbitrary connections model. The state
of the environment is modeled as a $k$-dimensional binary target vector $\bV$. Each
position in the vector may represent the presence of a target in an actual region in
space, or may have other interpretations, such as the presence of a specific chemical in
a sample. The possible target vectors are denoted $\bVi$, $i\in\{1,\ldots,2^k\}$. We say
that `a certain $\bV$ has occurred' if that vector represents the true state. We define a
sensor network $s(k,n)$ as a graph showing the connections of $n$ sensors to $k$
positions in the target vector. The sensor network has $n$ identical sensors. Sensor
$\ell$ senses exactly $c$ out of the $k$ spatial positions (shown in the graph as $c$
connections). We refer to such sensors as having a range $c$. Ideally, each sensor
produces a  value $x \in {\cal X}$ that is an {\em arbitrary function of the targets}
which it senses, $x_{\ell}=\Psi(v_{\ell t_1},\ldots,v_{\ell t_c})$. Thus, the `ideal
output vector'  of the sensor network $\bx$ depends on the sensor connections, and on the
target vector $\bV$ that occurs. However, we assume that each sensor output $y
\in\mathcal{Y}$ is corrupted by noise, so that the conditional p.m.f. $P_{Y|X}(y|x)$
determines the output. Since the sensors are identical, $P_{Y|X}$ is the same for all the
sensors. Further, we assume that the noise is independent in the sensors, so that the
`sensor output vector' $\by$ relates to the ideal output $\bx$ as
$P_{\byn|\bxn}(\by|\bx)=\prod_{\ell=1}^{n}P_{Y|X}(y_{\ell}|x_{\ell})$. Given the noise
corrupted output $\by$  of the sensor network, we detect the target vector $\bV$ which
occurred by using a detector $g(\by)$. Because of the constrained encoding of a sensor
network, we allow the decoder a distortion of $D\in[0,1]$. Denoting
$d_{\text{H}}(\bVi,\bVj)$ as the Hamming distance between two target vectors, the
tolerable distortion region of $\bVi$ is
$\cd_{\bVi}=\{j:\frac{1}{k}d_{\text{H}}(\bVi,\bVj) < D\}$. Given that $\bVi$ occurred,
the detector is in error if $g(\by)\not\in\cd_{\bVi}$.

Figure \ref{fig.ind_model} shows the target vector $\bV=(0,0,1,0,1,1,0)$ indicating $3$
targets among the $7$ target positions. In this example, the sensing function $\Psi$ is a
sum that indicates the number of positions which contain a target, $x_{\ell}=\sum_{u=1}^c
v_{\ell t_u}$, so that $x \in {\cal X} = \{0,1,\ldots,c\}$. Such a function could model a
chemical sensor that is sensitive to a subset of chemicals and whose output is linearly
proportional to the number of such chemicals present in the sample. More complex, e.g.
nonlinear, relationships between chemicals and sensor output require a different choice
of $\Psi$. In the figure,each sensor senses two target positions, and the sensors encode
the target vector as $\bx=(1,0,2,1)$. However, due to noise, the observed vector of
sensor outputs is $\by=(1,1,2,1)$. The target vector $\bV'=(0,1,1,0,1,1,0)$, which
differs from $\bV$  in one target position, is encoded as $\bx=(1,1,2,1)$. As a result a
detection algorithm can easily confuse $\bV'$ for $\bV$, demonstrating the limitation
imposed by the constrained encoding of a sensor network.

The arbitrary connections model describes large-scale detection problems that do not have
a spatial aspect. Examples of such applications include the detection of complex chemical
and computer network monitoring. Disease detection in a population where individual
sample can be combined is another such application. In addition to practical utility,
this model is easy to analyze and provides useful insights into large-scale detection
problems.

\subsection{Sensing Capacity Definitions}
\label{sec:itw_defs}

How many sensor measurements must a sensor network collect to detect the a target vector
to within a desired distortion? To answer this question we define the idea of a `sensing
capacity.' The probability of error of a sensor network given a that target vector $\bVi$
occurred is $P_{e,i,s}=\text{Pr}(\text{error}|i,s,\bxi, \by )=\text{Pr}(g(\by)
\not\in\cd_{\bVi}|\bVi,s,\bxi,\by )$. The expected probability of error for a sensor
network is $P_{e,s}= \sum_i P_{e,i,s}P_{\bVcap}(\bVi)$. The rate $R$ of a sensor network
is defined as the ratio of target positions being sensed to the number of sensor
measurements, $R=\frac{k}{n}$.  The \emph{sensing capacity} of a sensor network, $C(D)$,
is defined as the maximum rate $R$ such that below this rate there exists a sequence of
sensor networks $s(\lceil nR\rceil,n)$  whose expected probability of error across all
target vector goes to zero with increasing $n$, that is, $P_{e,s}\rightarrow 0$ as
$n\rightarrow \infty$ at a fixed rate $R$.

Is $C(D)$ nonzero? One of the main contributions of the theorem presented in this section
is to demonstrate that the sensing capacity can be strictly positive for the arbitrary
connections model. We use a random coding argument to obtain a bound on the sensing
capacity for the arbitrary connections model. Instead of constructing a sequence of
sensor network directly, we bound the expected probability of error, averaged over a
randomly generated ensemble of sensor networks. The sensor networks are generated as
follows. Each sensor connects to $c$ randomly chosen target positions out of the $k$
possible positions. The connections are made independently, and are chosen with
replacement. Therefore a sensor can choose the same target position more than once. When
we take the expectation over all such randomly generated sensor networks, the ideal
sensor outputs associated with each target vector become random. Since a sensor network
produces a codeword that is a function of the target vector, codeword distribution
depends on the occurring target vector. We denote the random vector which occurs when
$\bVi$ is the target vector as $\bxin$. Because each sensor forms its connections
independently,$P_{\bxin}(\bxi)=\prod_{\ell=1}^{n} P_{X_i}(x_{i\ell})$. It is important to
note that sensor outputs are in general not independent, and are only independent when we
condition on the occurrence of a particular target vector. Further, it is important to
note that the random vectors $\bxin$ and $\bxjn$, associated with a {\em pair of target
vectors}, $\bVi$ and $\bVj$ respectively, are {\em not independent}, since the sensor
connections produce a dependency between them. Thus, the `codewords' $\{\bxin,
i=1,2,\ldots,2^k\}$ of the sensor network are non-identical and dependent on each other,
unlike channel codes in classical information theory. Using this probabilistic model for
sensor network generation, we write the expected probability of error, averaged over the
sensor network ensemble as $P_{e}=\E_{S} [P_{e,S}]$. Theorem \ref{theorem1} in Section
\ref{sec.arbitthm} bounds this quantity to prove a lower bound $C_{LB}(D)$ on the sensing
capacity $C(D)$.

The statement of the result presented in this section relies on the method of types
\cite{Csiszar98}, and requires an explanation of {\em types} and {\em joint types}. Since
in the random sensor network construction each sensor connects to $c$ target positions
independently, the distribution of a sensor's ideal output $X_{i}$ depends only on the
type $\bgamma =(\gamma_0 , \gamma_1 )$ of $\bVi$. The type of target vector $\bVi$ is a
histogram of the number of 0's and 1's in $\bVi$. Here, $\gamma_0 $ denotes the fraction
of zeros in $\bVi$. Since sensor connections are generated independently and uniformly
across target positions in the arbitrary connections model we can write,
$P_{\bxin}(\bxi)= P^{\bgamma,n}(\bxi) = \prod_{\ell=1}^{n} P^{\bgamma }(x_{i\ell})$ for
all $\bVi$ of the same type $\bgamma$.

Since a single sensor network encodes all target vectors, pairs of codewords are
dependent, unlike codes in communications. The joint probability of two codewords
$P_{\bxin\bxjn}$ depends on the {\em joint type} of the target vectors $\bVi$ and $\bVj$.
The joint type is $\blambda = (\lambda_{00},\lambda_{01},\lambda_{10},\lambda_{11})$.
Here, $\lambda_{01}$ is the fraction of positions in $\bVi,\bVj$ where $\bVi$ has bit `0'
while $\bVj$ has bit `1'. Similarly, we define $\lambda_{00},\lambda_{10},\lambda_{11}$.

Following the notation introduced in \cite{Csiszar98}, $\blambda \ \in \
{\cp}_k(\{0,1\}^2)$, indicating that $\blambda$ is in the set of joint types of $k$-bit
binary vector pairs. Again, since sensor connections are generated independently and with
uniform probability across target positions, $P_{\bxin,\bxjn}(\bxi,\bxj) =
P^{\blambda,n}(\bxi,\bxj) = \prod_{\ell=1}^{n} P^{\blambda}(x_{i\ell},x_{j\ell})$ for all
$\bVi,\bVj$ of the same joint type $\blambda$. Since the joint type $\blambda$ also
defines the type $\bgamma$ of $\bVi$, we have $\lambda_{00}+\lambda_{01}=\gamma_0 $,
$\lambda_{10}+\lambda_{11}=\gamma_1$.

We give specific examples of these quantities using the example shown in Figure
\ref{fig.ind_model} where $c=2$ and sensors count the number of targets present in the
target positions that they sense. Table \ref{tab.joint} lists the types of four vectors
$\bVj$, and their joint type with the target vector in the example $\bVi=0010110$. Given
a target vector, a sensor will output `2' only if both of its connections connect to
positions with a `1.' For a vector of type $\bgamma$, this occurs with probability
$(\bgamma_1)^2$. Table \ref{tab.XiTable} describes the complete output p.m.f. for a
randomly generated sensor, given that a vector of type $\bgamma$ occurred. Given two
target vectors $\bVi,\bVj$ of joint type $\blambda$, a sensor will output `0' for both
target vectors only if both its connections are connected to target positions that have a
`0' bit in both these target vectors. This happens with probability $(\lambda_{00})^2$.
Table \ref{tab.jointTable} lists the complete joint p.m.f. $P_{X_iX_j}(x_{i},x_{j})=
P^{\blambda}(x_{i},x_{j})$ of a randomly generated sensor for two target vectors with a
joint type $\blambda$.

\subsection{Sensing Capacity Lower Bound}
\label{sec.arbitthm}

We specify two probability distributions which we will utilize in the main theorem of
this section. The first is the joint distribution of the ideal output $\bxi$ when $\bVi$
occurs, and the noise corrupted output $\by$, $P_{\bxin\byn}(\bxi,\by)=
\prod_{\ell=1}^{n}P_{X_iY}(x_{i\ell},y_{\ell}) =
\prod_{\ell=1}^{n}P_{X_i}(x_{i\ell})P_{Y|X}(y_{\ell}|x_{i\ell})$. The second distribution
is the joint distribution of the ideal output $\bxi$ corresponding to $\bVi$ and the
noise corrupted output $\by$ generated by the occurrence of a {\em different} target
vector $\bVj$. We can write this joint distribution as $Q_{\bxin\byn}^{(j)}(\bxi,\by)=
\prod_{\ell=1}^{n}Q_{X_iY}^{(j)}(x_{i\ell},y_{\ell}) = \prod_{\ell=1}^{n}\sum_{a \in
{\cx}}P_{X_i,X_j}(x_{i\ell},x_j=a)P_{Y|X}(y_{\ell}|x_j=a)$. Note that although $\byn$ was
produced by $\bxjn$, $\bxin$ and $\byn$ are dependent because of the dependence of
$\bxin$ and $\bxjn$.

We argued earlier that due to the random sensor network construction, $P_{\bxin}$ and
$P_{\bxin\bxjn}$ can be compute using the type $\bgamma$ of $\bVi$ and joint type
$\blambda$ of $\bVi,\bVj$ respectively. Thus, we write $P_{\bxin\byn}(\bxi,\by) =
\prod_{\ell=1}^{n} P_{{X_{i}Y}}^{\bgamma}(x_{i\ell},y_{\ell})$ where
$P_{{X_{i}Y}}^{\bgamma}(x_{i},y) = P^{\bgamma}(x_{i})P_{Y|X}(y|x_{i})$. Similarly, we
write $Q_{\bxin\byn}^{(j)}(\bxi,\by)=\prod_{\ell=1}^{n}
Q_{{X_{i}Y}}^{\blambda}(x_{i\ell},y_{\ell})$ where $Q_{{X_{i}Y}}^{\blambda}(x_{i},y) =
\sum_{a \in {\cx}} P^{\blambda}(x_{i},x_j=a) P_{Y|X}(y|x_j=a)$. We can now ready to state
the main theorem of this section.

\begin{theorem}[Sensing Capacity Theorem for the Arbitrary Connections Model] \label{theorem1}
The sensing capacity at distortion $D$ is bounded as,
\begin{equation}
C(D) \geq C_{LB}(D) =
  \min_{\begin{subarray}{c}\blambda \\ \lambda_{01}+\lambda_{10}\geq D \\
\lambda_{00}+\lambda_{01}=\gamma_{0} \\
\lambda_{10}+\lambda_{11}=\gamma_{1}
\end{subarray}}\frac{D\left(P_{X_iY}^{\bgamma}\|Q_{{X_{i}Y}}^{\blambda}\right)}{H(\blambda)-H(\bgamma)}
\end{equation}
where $\bgamma=(0.5,0.5)$ and
$\blambda=(\lambda_{00},\lambda_{01},\lambda_{10},\lambda_{11})$ is an arbitrary
probability mass functions.
\end{theorem}

The most striking difference between the result shown in Theorem \ref{theorem1}, and
Shannon's channel capacity results is that the bound on the sensing capacity is not a
mutual information. From the definition of $Q_{{X_{i}Y}}^{\blambda}$, we notice that if
the `codewords' $\bxin$ were independent, the Kullback-Leibler distance would reduce to
the mutual information between $X_i$ and its noisy version $Y$. This is an important
difference because of the frequent use of mutual information as a sensor selection metric
(e.g. \cite{DurrantWhyte94}), and indicates that the mutual information is not the
correct notion of information for large-scale detection applications. The difference
between channel capacity and sensing capacity arises due to different codeword
geometries. In proofs of the achievability of channel capacity, since a codeword can be
arbitrarily assigned to a message in communications, codewords are distributed uniformly.
In a sensor network, the codeword distribution depends on the state of the environment
(the target vector). Codewords are clustered, with similar target vectors encoded as
similar codewords.  As a result, similar target vectors are more likely to be confused
due to noise than dissimilar target vectors. The Kullback-Leibler distance in Theorem
\ref{theorem1} is the appropriate information measure for such a codeword geometry. The
denominator in Theorem (\ref{theorem1}) accounts for disparities in the size of codewords
clusters. The minimization over the joint type appears because the ``closest'' target
vectors dominate the error probability. Thus, the sensing capacity is similar to
classical channel capacity, with differences arising due to the non-identical, dependent
codeword distribution.

The proof of Theorem \ref{theorem1} broadly follows the proof of channel capacity
provided by Gallager \cite{Gallager68}, by analyzing a union bound of pair-wise error
probabilities, averaged over randomly generated sensor networks. However, it differs from
\cite{Gallager68} in several important ways. In our sensor network model, the codewords
are dependent on each and non-identically distributed. To prove our bound, we group the
exponential number of pair-wise error terms into a polynomial number of terms using the
method of types.

\begin{proof}

%MAP
%EQUALLY LIKELY
We assume a maximum-likelihood detector $g_{\text{ML}}(\by)=\argmax_j
P_{\byn|\bxn}(\by|\bxj)$. For this detector, we consider $P_{e}=\frac{1}{2^k}\sum_i \
P_{e,i}$, where we assume that the target vectors are equally likely. $P_{e,i}$ is the
error probability when the $i^th$ target vector occurs, averaged over all randomly
generated sensor networks. For a fixed sensor network $s$ there is a known and fixed
correspondence between target vectors $\bVi$ and codewords $\bxi$. Since our sensor
network is chosen randomly, the set of codewords is random,
$\cac=\{\bxn_1,\ldots,\bxn_{2^{k}}\}$.
\begin{equation} \label{Pe1}
P_{e} = \E_{\bVcap\byn\cac}\left[ \text{Pr}(g(\byn) \not\in\cd_{\bVcap}|\bVcap,\cac,\byn
) \right]
\end{equation}
Using the fact that we are taking the expectation of a probability, we bound $P_e$ as
follows,
\begin{equation} \label{Pe2}
P_{e} \leq \E_{\bVcap\byn\cac}\left[\sum_w \text{Pr}(g(\byn) \in\cs_w|\bVcap,\cac,\byn
)^{\rho} \right]
\end{equation}
where $\rho\in [0,1]$, and $\{\cs_1,\cs_2,\ldots\}$ is a partition of the complement of
$\cd_{\bVcap}$, denoted $\cd_{\bVcap}^C$. Using the union bound, we upper bound the
probability $\text{Pr}(g(\byn) \in\cs_w|\bVcap,\cac,\byn )$ as follows,
\begin{equation} \label{Pe3}
P_{e} \leq \E_{\bVcap\byn\cac}\left[\sum_w \left( \sum_{j \in \cs_w}
\text{Pr}(g(\byn)=j|\bVcap,\cac,\byn )\right)^{\rho} \right]
\end{equation}
The term $\text{Pr}(g(\byn)=j|\bVcap,\cac,\byn )$ is a pairwise error term that depends
only on the codewords $\bxin$ and $\bxjn$. Using this observation, the fact that
$x^{\rho}$ is a concave function for $\rho \in [0,1]$, and Jensen's inequality, we
obtain,
\begin{equation} \label{Pe4}
P_{e} \leq \E_{\bVcap\byn\bxin}\left[\sum_w \left( \sum_{j \in \cs_w}
\E_{\bxjn|\bxin}\left[\text{Pr}(g(\byn)=\bVj|\bVcap,\bxin,\bxjn,\byn )\right]
\right)^{\rho} \right]
\end{equation}
The term $\text{Pr}(g(\byn)=\bVj|\bVcap,\bxin,\bxjn,\byn )$ is a one zero function,
equaling one when $g(\byn)=\bVj$ and zero otherwise. Using our assumption that $g$ is an
ML detector we upper bound this probability as follows,
\begin{multline} \label{Pe5}
P_{e} \leq \frac{1}{2^k}
\sum_{i}\sum_{\bxi\in\cx^n}\sum_{\by\in\cy^n}P_{\bxin}(\bxi)P_{\byn|\bxn}(\by|\bxi)
\\ \cdot \sum_{w} \left(\sum_{j\in\cs_w}\sum_{\bxj\in\cx^n} P_{\bxjn|\bxin}(\bxj|\bxi)\left( \frac{P_{\byn|\bxn}(\by|\bxj)}{P_{\byn|\bxn}(\by|\bxi)}\right)^{\frac{1}{1+\rho}} \right)^{\rho}
\end{multline}
The bound in equation (\ref{Pe5}) has an exponentially large number of terms. Earlier in
this paper, it was shown that the distributions in this bound can be completely specified
by the type $\bgamma$ and joint type $\blambda$ rather than the specific $i,j$ pair of
target vectors. To do this, we choose each $\cs_w$ to be a distinct joint type
$\blambda$, and let $w$ index the set $S_{\bgamma}(D)$ of all $\blambda$ that are the
joint type of $\bVi$ and $\bVj\in \cd_{\bVi}^C$. We group the summation over $i$
according to the type of $\bVi$. Grouping according to the type and joint type enables us
to take advantage of the fact that the number of types is polynomial in $k$. After
grouping according to types, we write equation (\ref{Pe5}) as,
\begin{multline} \label{Pe6}
P_{e} \leq \frac{1}{2^k}
\sum_{\bgamma}\alpha(\bgamma,k)\sum_{\bxi\in\cx^n}\sum_{\by\in\cy^n}P^{\bgamma,n}(\bxi)P_{\byn|\bxn}(\by|\bxi)
\\ \cdot \sum_{\blambda \in S_{\bgamma}(D)} \left( \beta(\blambda,k)\sum_{\bxj\in\cx^n} P^{\blambda,n}(\bxj|\bxi)\left( \frac{P_{\byn|\bxn}(\by|\bxj)}{P_{\byn|\bxn}(\by|\bxi)}\right)^{\frac{1}{1+\rho}} \right)^{\rho}
\end{multline}
where $\alpha(\bgamma,k)$ is the number of target vectors $\bVi$ of length $k$ and type
$\bgamma$, and where $\beta(\blambda,k)$ is the number of target vectors $\bVj$ of length
$k$ and joint type $\blambda$ with a target vector $\bVi$ of type $\bgamma$.
$S_{\bgamma}(D)$ is defined as
\begin{equation} \label{eqn.defSDarbit}
S_{\bgamma}(D) = \{\blambda: \lambda_{01}+\lambda_{10} \geq D,\
\lambda_{00}+\lambda_{01}=\bgamma_{0} , \ \lambda_{10}+\lambda_{11}= \bgamma_{1}\}
\end{equation}
Using standard results from the method of types \cite{Csiszar98} about the number of
binary vectors of a given type, we obtain the bound, $\alpha(\bgamma,k)\leq 2^{k
H(\bgamma)}$. The number of vectors with a given joint type is bounded as,
\begin{equation}
\beta(\blambda,k) = \binom{k\gamma_0}{k\lambda_{00}}\binom{k\gamma_1}{k\lambda_{11}} \leq
 2^{k(H(\blambda)-H(\bgamma))}
\label{eqn.jointcount_ind}
\end{equation}
Combining equation (\ref{Pe6}) with the bounds on $\alpha$ and $\beta$, and using the
conditional independence of sensor outputs, we obtain,
\begin{equation}  \label{Pe7}
P_{e}   \leq \sum_{\bgamma}\sum_{\blambda \in
S_{\bgamma}(D)}2^{-k(1-H(\bgamma))}2^{k\rho(H(\blambda)-H(\bgamma))}
2^{-nE(\rho,\blambda)}
\end{equation}
where $E(\rho,\blambda)$ is defined as below,
\begin{equation}
E(\rho,\blambda)=  -\log\Bigg(\sum_{a_i\in\cx} \sum_{b\in\cy}P^{\bgamma_i }(a_i)
P_{Y|X}(b|a_i)^{\frac{1}{1+\rho}}\Big(\sum_{a_j\in\cx}P^{\blambda}(a_j|a_i)P_{Y|X}(b|a_j)^{\frac{1}{1+\rho}}\Big)^{\rho}\Bigg)
\end{equation}
Since the number of types $\bgamma$ and joint types $\blambda$ are upper bounded by
$(k+1)^2$ and $(k+1)^4$ respectively, and $k=\lceil nR\rceil$, implying $k < nR+1$,
(\ref{Pe7}) is bounded as,
\begin{equation}
P_{e} \leq 2^{-n(o_1(n)+E_r(R,D))}
\end{equation}
where $o_1(n)\rightarrow 0$ as $n \rightarrow \infty$, and where $E_r(R,D)$ is defined
as,
\begin{equation}
E_r(R,D)  =  \min_{\bgamma}  \min_{\blambda \in S_{\bgamma}(D)} \max_{0 \leq \rho \leq 1}
\left(E(\rho,\blambda)+R(1-H(\bgamma)) -  \rho R (H(\blambda)-H(\bgamma))\right)
\end{equation}
The average error probability $P_e \rightarrow 0$ as $n \rightarrow \infty$ if $E_r(R,D)
>0$. Observing that $E(0,\blambda)=0 \ \forall \blambda$, we let $\rho$ go to zero,
rather than optimizing it, thus resulting in a lower bound on $E_r(R,D)$. In the above
expression, this implies that in order for $R$ to be achievable
$\frac{E(\rho,\blambda)}{\rho}+R\frac{1-H(\bgamma)}{\rho}-R(H(\blambda)-H(\bgamma))$ must
be positive for all types and joint types as $\rho \rightarrow 0$.

For $H(\bgamma) \neq 1$, $\frac{1-H(\bgamma)}{\rho}\rightarrow \infty$ as
$\rho\rightarrow 0$. For such a $\bgamma$, $P_e\rightarrow 0$ since $E_r(R,D)$ is
positive for all rates $R$. Since we seek to bound $R$ for which $E_r(R,D)$ is positive
for all types and joint types, we let $\bgamma=(0.5,0.5)$. This implies that as
$\rho\rightarrow 0$, $R$ is achievable when the derivative of $E(\rho,\blambda)$ with
respect to $\rho$ at $\rho=0$ is greater than $R(H(\blambda)-H(\bgamma))$. It can be
easily shown that, $\partial E(\rho,\blambda)/\partial \rho\big|_{\rho=0}=
D(P_{X_iY}^{\bgamma}\|Q_{X_iY}^{\blambda})$. Using this derivative in the analysis above,
we see that the achievable rates $R$ are bounded as below.
\begin{equation} \label{eqn.ITWRBound}
R \leq \min_{\begin{subarray}{c}\blambda \\ \lambda_{01}+\lambda_{10}>D \\
\lambda_{00}+\lambda_{01}=\gamma_0 \\
\lambda_{10}+\lambda_{11}=\gamma_1
\end{subarray}}\frac{D\left(P_{X_iY}^{\bgamma}\|Q_{{X_{i}Y}}^{\blambda}\right)}{H(\blambda)-H(\bgamma)}
%\vspace{-0.1in}
\end{equation}
where $\bgamma=(0.5,0.5)$, and $\blambda$ is an arbitrary p.m.f. since $n\rightarrow
\infty$. Therefore, the right hand side of (\ref{eqn.ITWRBound}) is a lower bound on
$C(D)$. \end{proof}

\subsection{Numerical Results}

We compute the capacity bound $C_{LB}(D)$ in (\ref{theorem1}) for various distortions,
noise levels, and sensor ranges. A sensor of range $c$ is connected to $c$ target
positions. We assume that the sensing function $\Psi$ simply counts the number of target
positions in the sensor range with a target present. The sensor noise model assumes that
the probability of counting error decays exponentially with the error magnitude. In the
figures, `Noise = $p$' indicates that for a sensor, $P(Y \ne X) =p$, with ${\cal Y}={\cal
X}$ assumed. In Figure \ref{plot.RD_indplot}, we demonstrate $C_{LB}(D)$ for various
sensor noise levels and ranges. We compute this bound by systematically sampling the
space of possible $\blambda$. While $\blambda$ is a four-dimensional vector, because of
constraints we need to sample only two dimensions in order search over all valid
$\blambda$. In all cases, $C_{LB}(D)$ approaches $0$ as $D$ approaches $0$. This occurs
because similar target vectors have similar codewords due to dependence in the codeword
distribution. The relative magnitude of the bounds for sensors of various $c$ and noise
levels describes tradeoffs among sensor types that can be captured by our result. Some
tradeoffs are intuitive. For example, lower noise sensor of range $c$ have a higher
sensing capacity than higher noise sensors of the same range. Other tradeoffs are more
complex. For example the tradeoff between shorter and longer range sensors depends on the
desired distortion. Sensors of range $4$ and noise $0.10$ result in a higher sensing
capacity than sensors of range $2$ and noise $0.01$ for distortion above $0.047$. The
opposite is true for distortions below $0.047$. Thus, the bound presented in
(\ref{theorem1}) describes a complex tradeoffs between sensor noise, sensor range, and
the desired detection accuracy.

Figure \ref{plot.RN_indplot} shows $C_{LB}(D)$ at $D=0.1$ as a function of sensor noise
level for sensors of various range and sensing functions. This figure demonstrates that
the strategy of simple sensor replication, which is a popular practical method for
reducing error probability, can be inefficient. For example, for sensors of range $4$ and
a sum sensing function, a rate of $0.61$ is achievable at noise level $0.1$. If each
sensor with noise $0.1$ is replicated three times and majority decoding is used, the
noise can be reduced to $3\times (0.1)^2\times 0.9+(0.1)^3 = 0.028$. For a noise level of
$0.028$, $C_{LB}(0.1)$ equals $0.91$ for a sensor of range $4$ and a sum sensing
function. However, due to sensor replication, the rate is reduced to $0.91/3=0.303$. This
rate is significantly lower than the rate of $0.61$ for sensors of noise $0.1$ achievable
by using our random sensor network construction. Thus, the bound indicates that
cooperative sensor strategies can require significantly fewer sensor measurements than
sensor replication. Figure \ref{plot.RN_indplot} also shows $C_{LB}(D)$ at $D=0.1$ for
sensors with $c=4$ and a weighted sum sensing function with weights $\{1,0.5,0.25,0.1\}$.
This sensing function has a higher sensing capacity than sensors with the same range and
an un-weighted sum sensing function across all noise levels. We conjecture that this
occurs because a weighted sum can distinguish among more target configurations than an
un-weighted sum. Interestingly, the gap between the two sensing functions increases with
increasing noise.

Using the loopy belief propagation algorithm \cite{Pearl88} we empirically examined
sensor probability of error as a function of rate. We generated sensor networks of
various rates by setting the number of targets, and varying the number of sensors. We
chose the number of connections to be $c=4$, the distortion level to be $0.1$, and the
noise level to be $0.1$ (i.e. $P(Y \ne X) =0.1$, with ${\cal Y}={\cal X}$). As in the
previous section, we assume that the probability of error decays exponentially with error
magnitude. We empirically evaluated the average error rate obtained in decoding target
vectors in a randomly generated set of sensor networks. We plotted the average error rate
for each rate value, and for various numbers of targets as shown in Figure \ref{bp_plot}.
As the number of targets increase, the transition from high error to low error rate
becomes increasingly sharp. However, all the error curves are well below the capacity
value $C_{LB}(0.1)=0.62$. We conjecture that this occurs because belief propagation is
suboptimal for graphs with cycles.

\subsection{Extensions}
\label{sec.arbit.extensions} Section \ref{sec.arbitraryconnsmodel} introduced a sensor
network model where each sensor is allowed to make arbitrary connections to the target
vector. In several situations, more complex sensor network models may be necessary. This
section describes extensions of the arbitrary connection model. Extensions that account
for contiguity in sensor connections require a new model and are discussed in Section
\ref{sec.ipsn}. The first extension considers non-binary target vectors. Binary target
vectors indicate the presence or absence of targets at the spatial positions. A target
vector over a general finite alphabet may indicate, in addition to the presence of
targets, the class of a target. Alternatively, the entries of non-binary vectors can
indicate levels of intensity or concentration. Assuming a non-binary target vector, we
can define types and joint types over an alphabet $\cv$, and apply the same analysis as
before to obtain the sensing capacity bound below. %\vspace{0.2in}
\begin{equation}
 C(D) \geq C_{LB}(D) =
 \min_{\begin{subarray}{c}\blambda \\ \sum_{a\neq b}\lambda_{ab}\geq D \\
\sum_{b}\lambda_{ab}=\gamma_a
\end{subarray}} \frac{D\left(P_{X_iY}^{\bgamma}\|Q_{{X_{i}Y}}^{\blambda}\right)}{H(\blambda)-H(\bgamma
)}
\end{equation}
\textit{where ${\bgamma}=\left(\gamma_{a}=\frac{1}{|\cv|},a\in \cv \right)$, while
$\blambda = (\lambda_{ab},\ a,b\in \cv )$ is an arbitrary probability mass function.}

The second extension allows the following a priori distribution over target vectors.
Assume that each target position is generated i.i.d. with probability $P_V$ over the
alphabet $\cv$. This may model the fact that targets are sparsely present. The previous
analysis can be extended to a Maximum-a-Posteriori (MAP) detector, instead of the ML
detector considered earlier, resulting in the following sensing capacity bound.
\begin{equation}
 C(D) \geq C_{LB}(D) =
 \min_{\begin{subarray}{c}\blambda \\ \sum_{a\neq b}\lambda_{ab} \geq D \\
\sum_{b}\lambda_{ab}=\gamma_{ia}
\end{subarray}} \frac{D\left(P_{X_iY}^{\bgamma_i}\|Q_{{X_{i}Y}}^{\blambda}\right)}{H(\blambda)-H(\bgamma
_j)-D(\bgamma_j\|P_V)}
\end{equation}
\textit{where ${\bgamma_i}=P_V$, $\blambda = (\lambda_{ab},\ a,b\in \cv )$ is an
arbitrary probability mass function and $\bgamma_j$ is the marginal of $\blambda$
calculated as $\gamma_{jb}=\sum_{a} \lambda_{ab}$.}

A third extension accounts for heterogenous sensors, where each class of sensor possibly
has a different range $c$, noise model $P_{Y|X}$, and/or sensing function $\Psi$. Let the
sensor of class $l$ be used with a given relative frequency $\alpha_l$. For such a model
the sensing capacity bound is as follows.
\begin{equation}
C(D) \geq C_{LB}(D) =  \min_{\begin{subarray}{c}\blambda \\ \sum_{a\neq b}\lambda_{ab} \geq D \\
\sum_{b}\lambda_{ab}=\gamma_{ia}
\end{subarray}} \frac{\sum_l \alpha_l D\left(P_{X_iY}^{\bgamma_i,l}\|Q_{{X_{i}Y}}^{\blambda,l}\right)}{H(\blambda)-H(\bgamma
_j)-D(\bgamma_j\|P_V)}
\end{equation}
\textit{where ${\bgamma_i}=P_V$, $\blambda = (\lambda_{ab},\ a,b\in \cv )$ is an
arbitrary probability mass function and $\bgamma_j$ is the marginal of $\blambda$
calculated as $\gamma_{jb}=\sum_{a} \lambda_{ab}$.}

\section{Sensing Capacity of Contiguous Connections Model} \label{sec.ipsn}

In this section, we analyze the sensing capacity of a sensor network model that models
contiguity in a sensor's connections. Figure \ref{fig:dep_model} shows an example of such
a model. Sensor $\ell$ is connected to exactly $c$ \emph{contiguous} positions out of the
$k$ spatial positions. In contrast, the arbitrary connections model analyzed in the
previous section did not account for localized sensor observations since each sensor
could sense any $c$ (not necessarily contiguous) spatial positions.

\subsection{Higher Order Types}

The statement of the result for contiguous models requires higher order types
\cite{Csiszar98}. We introduce {\em circular c-order types} and {\em circular c-order
joint types}. We define the circular $c$-order type of a binary sequence (i.e. a target
vector) as a $2^c$ dimensional vector, $\bgamma$, where each entry in the vector
corresponds to the frequency of occurrence of one of the possible subsequences of length
$c$. A circular sequence is one in which the last element of the sequence precedes the
first element of the sequence. The total number of subsequences of length $c$ that can
occur in a circular sequence of length $k$ is $k$. For example, for a binary target
vector and $c=2$, $\bgamma =(\gamma_{00}, \gamma_{01},\gamma_{10}, \gamma_{11} )$. While
it is possible to prove our bound using non-circular types as shown in
\cite{RachlinIPSN05}, circular types lead to the same asymptotic result with the benefit
of significantly simpler notation. The notational simplicity arises out of the fact that
the lower order circular types are precise marginals of the higher order circular types.
Although all the types in this section are circular, we will omit the word `circular'
when referring to types in the remainder of this section for brevity.

We denote the set of all $c$-order types over the alphabet $\cv^c$ for target vectors of
length $k$ as $\cp_k(\cv^c)$. Since each sensor independently chooses a block of $c$
contiguous spatial positions, the distribution of its ideal output $X_{i}$ depends only
on the $c$-order type $\bgamma$ of the target vector $\bVi$ which occurs. For a sensing
function $\Psi$ and a target vector $\bVi$ of type $\bgamma$,
\begin{equation}
P_{X_i}(X_i=x)=\sum_{\begin{subarray}{c}\ba\in\cv^c\\\Psi(\ba)  =
x\end{subarray}}\gamma_{\ba} \doteq  P^{\bgamma }(x)
\end{equation}
%Thus, $P_{\bxin}(\bxi)= P^{\bgamma ,n}(\bxi) = \prod_{\ell=1}^{n} P^{\bgamma }(x_{i\ell})$ for all
%$\bVi$ of type $\bgamma$.

Next, we note that the joint distribution $P_{\bxin\bxjn}$ depends on the {\em c-order
joint type} $\blambda$ of the $i^{th}$ and $j^{th}$ target vectors $\bVi,\bVj$.
$\blambda$ is the vector of $\lambda_{(\ba)(\bb)}$, the fraction of positions in
$\bVi,\bVj$ where $\bVi$ has a bit subsequence $\ba$ while $\bVj$ has a bit subsequence
$\bb$. For example, when $c=2$ and $\cv=\{0,1\}$,
$\blambda=(\lambda_{(00)(00)},\ldots,\lambda_{(11)(11)})$. We denote the set of all
c-order joint types over the alphabet $\cv^{c}$ for target vectors of length $k$ as
$\cp_k(\cv^{c},\cv^{c})$. Each $\blambda \in \cp_k(\cv^{c},\cv^{c})$ must satisfy the
normalization constraint that the sum over all entries of $\lambda$ equals one. Since the
joint type $\blambda$ also defines the type $\bgamma$ of $\bVi$, for all
$\{\ba\}\in\cv^c$ we must have $\gamma_{\ba} =\sum_{\bb\in\cv^c}\lambda_{(\ba)(\bb)}$.
Taking advantage of the fact that for circular types, lower order types are precise
marginals of higher order types, we denote $\lambda_{(a)(b)}=\sum_{\ba'\in\cv^{c-1}}
\sum_{\bb'\in\cv^{c-1}}\lambda_{(a\ba')(b\bb')}$. $\lambda_{(a)(b)}$ is the normalized
count of locations where target vector $i$ has value $a$ while target vector $j$ has
value $b$. Since each sensor depends only on the $c$ contiguous targets bits which it
senses, $P_{\bxin,\bxjn}$ depends only on the joint type $\blambda$. For target vectors
$\bVi$,$\bVj$ of c-order joint type $\blambda$,
\begin{equation}
P_{X_iX_j}(X_i=x_i,X_j=x_j)= \sum_{\begin{subarray}{c}\ba,\bb\in\cv^c\\
\Psi(\ba)=x_i,\  \Psi(\bb)=x_j
\end{subarray}}\lambda_{(\ba)(\bb)} \doteq P^{\blambda}(x_{i},x_{j}) %= P^{\blambda}(x_{j} | x_{i})P^{\bgamma }(x_i)
\end{equation}
%Thus, $P_{\bxjn|\bxin}(\bxj | \bxi)  = P^{\blambda,n}(\bxj | \bxi) = \prod_{\ell=1}^{n}
%P^{\blambda}(x_{j\ell} | x_{i\ell})$ for all $i,j$ of the same joint type $\blambda$.

For example, for binary target vectors and $c=2$, vectors $00000000, 01101000, 01000111$
have $\bgamma = (1,0,0,0), (3/8,2/8,2/8,1/8), (2/8,2/8,2/8,2/8)$ respectively. Table
\ref{tab:lambda} contains the 2-order joint type of two target vectors. Consider a sensor
network where each sensor is randomly connected to $c=2$ contiguous spatial positions. We
assume that $\Psi$ outputs the number of targets which the sensor observes. Thus, each
sensor has an ideal output alphabet $\cx=\{0,1,2\}$. For target vectors of type
$\bgamma$, $P(X_i=0) =\gamma_{00}, P(X_i=1)=\gamma_{01}+\gamma_{10},P(X_i=2)=\gamma_{11}$
respectively. Given two target vectors $\bVi,\bVj$ of joint type $\blambda$, a sensor
will output `0' for both target vectors only if both of its connections see a `0' bit in
both target vectors. This happens with probability $\lambda_{(00)(00)}$. Table
\ref{tab:XiXj} lists the joint p.m.f. $P_{X_iX_j}(x_{i},x_{j})= P^{\blambda}( x_{i},
x_{j} )$ for all output pairs $x_i,x_j$ corresponding to joint type $\blambda$. The table
shows that $X_i,X_j$ are not independent, in general.

To prove Theorem \ref{theorem1}, we bounded the number of target vectors $\bVj$ that have
a given joint type with a target vector $\bVi$ in equation (\ref{eqn.jointcount_ind}). To
prove a sensing capacity bound for the contiguous connections model we prove a bound on
the number of target vectors $\bVj$ that have a joint c-order type $\blambda$ with a
target vector of $c$-order type $\bgamma$ in the lemma below. Before proceeding, we
introduce the following notation. The set of length $k$ target vectors of $c$-order type
$\bgamma$ is denoted $\ct_{\bgamma}^{k}$. The set of pairs of length $k$ target vectors
of joint type $\blambda$ is denoted $\ct_{\blambda}^{k}$. The set of length $k$ target
vectors that have joint c-order type $\blambda$ with a given vector of type $\bgamma$, is
denoted $\ct_{\blambda|\bgamma}^{k}$.

\begin{lemma}[Bound on $|\ct_{\blambda|\bgamma}^{k}|$ ] \label{lemma:countingtypes}
The number of binary vectors of length $k$ with c-order joint type $\blambda$ for a given
vector of $c$-order type $\bgamma$, denoted $|\ct_{\blambda|\bgamma}^{k}|$, is bounded as
follows
\begin{equation}
|\ct_{\blambda|\bgamma}^{k}|\leq C(k)2^{k(H(\tblambda|\blambda')-H(\tbgamma|\bgamma'))}
\end{equation}
$C(k)=2^{2(c-1)}k^{2^{c-1}}(k+1)^{2^{2(c-1)}}$ and
$\blambda'=\{\lambda_{(\ba)(\bb)},\forall \ba,\bb \in \cv^{c-1}\}$ is a probability mass
function defined as $\lambda_{(\ba)(\bb)}'=\sum_{a,b\in\cv}\lambda_{(\ba a)(\bb b)}$.
$\tblambda=\{\tlambda_{(\ba a)(\bb b)},\forall\ba,\bb \in \cv^{c-1},\forall a,b\in\cv \}$
is a conditional probability mass function defined as $\tlambda_{(\ba a)(\bb
b)}=\frac{\lambda_{(\ba a)(\bb b)}}{\lambda_{(\ba)(\bb)}}$. $\bgamma'=\{\gamma_{\ba}',
\forall\ba\in\cv^{c-1}\}$ is probability mass function defined as
$\gamma_{\ba}'=\sum_{a\in\cv}\gamma_{\ba a}$. $\tbgamma=\{\tgamma_{\ba
a},\forall\ba\in\cv^{c-1},\forall a\in\cv\}$ is a conditional probability mass function
defined as $\tgamma_{\ba a}=\frac{\gamma_{\ba a}}{\gamma_{\ba}}$.
\end{lemma}

\begin{proof}
To bound $|\ct_{\blambda|\bgamma}^{k}|$, we begin by bounding $|\ct_{\bgamma}^{k}|$. The
$c$-order type  $\bgamma$ of a binary vector specifies a 2-order type (referred to as a
Markov type) of a vector who entries are in an alphabet of cardinality $2^{c-1}$.
Consider the vector $011010001$. Denoting a pair of bits using $\{0,1,2,3\}$, as we move
from left to right over this vector, one bit at a time, the sequence obtained is
$132120012$. The 3-order type $\bgamma$ specifies the 2-order type over this new vector.
For example, the fraction of times $1$ transitions to $3$ is equal to $\gamma_{011}$ and
the fraction of times $1$ transitions to $2$ is equal to $\gamma_{010}$. Any c-order type
over a binary sequence can thus be mapped to a 2-order type over a sequence with symbols
in an alphabet of cardinality $2^{c-1}$. \cite{Davisson81} proves bounds on the number of
sequences that correspond to a 2-order circular type over a sequence with an alphabet
$\cv$. Given our mapping from a c-order type to a 2-order type, we can apply this result
to obtain the following bound,
\begin{equation} \label{eqn:MarkovTypeCount}
|\ct_{\bgamma}^{k}|\geq
C_1(k)2^{k(H(\bgamma)-H(\bgamma'))}=C_1(k)2^{kH(\tbgamma|\bgamma')}
\end{equation}
where $C_1(k)=k^{-2^{c-1}}(k+1)^{-2^{2(c-1)}}$. We now bound $|\ct_{\blambda}^{k}|$ using
a similar argument. The c-order joint type $\blambda$ of a pair of binary vectors
specifies a 2-order type of a single vector whose entries are symbols from an alphabet of
cardinality $2^{2(c-1)}$. We consider an example for a 3-order joint type, and for
vectors $\bV=011010001$ and $\bV'=101011011$. We can rewrite these vectors as a single
vector whose entries at location $i$ are defined by the pair of entries $v_i,v_i'$ and
the subsequent pair of entries $v_{i+1},v_{i+1}'$. These four entries, combined as $v_i
v_{i+1} v_i' v_{i+1}'$ are mapped to a symbol in an alphabet of cardinality $2^4$ by
reading the entries as a binary number (i.e. $0000=0$, $0001=1$, $\ldots$). In this
manner, $\bV$ and $\bV'$ are mapped to a vector $(6,14,10, 9,11,2,1,7,11)$. The 3-order
joint type $\blambda$ specifies the 2-order type over this new vector. For example, the
fraction of times $1$ transitions to $7$ is equal to $\lambda_{(001)(011)}$ and the
fraction of times $2$ transitions to $1$ is equal to $\lambda_{(000)(101)}$. Any c-order
joint type over a binary sequence can thus be mapped to a 2-order type over a sequence
with symbols in an alphabet of cardinality $2^{2(c-1)}$. We use the results of
\cite{Davisson81} again. Given our mapping from a c-order joint type to a 2-order type,
we can apply this result to obtain the following bound,
\begin{equation} \label{eqn:MarkovJointTypeCount}
|\ct_{\blambda}^{k}|\leq C_2 2^{k(H(\blambda)-H(\blambda'))}=C_2
2^{kH(\tblambda|\blambda')}
\end{equation}
where $C_2=2^{2(c-1)}$. We observe that $|\ct_{\blambda|\bgamma}^{k}|$ depends only on
the type $\bgamma$ of the vector on which we are conditioning, and not on the actual
vector. Therefore,
$|\ct_{\blambda|\bgamma}^{k}|=\frac{|\ct_{\blambda}^{k}|}{|\ct_{\bgamma}^{k}|}$. Using
equations (\ref{eqn:MarkovTypeCount}) and (\ref{eqn:MarkovJointTypeCount}), we obtain the
following bound,
\begin{equation}
|\ct_{\blambda|\bgamma}^{k}|\leq C(k)2^{k(H(\tblambda|\blambda')-H(\tbgamma|\bgamma'))}
\end{equation}

where $C(k)=C_{1}^{-1}(k)C_2$

\end{proof}

\subsection{Sensing Capacity Lower Bound}

We define $P_{{X_{i}Y}}^{\bgamma}$ $Q_{{X_{i}Y}}^{\blambda}$ as they were defined for the
arbitrary connections model bounds, with the only difference arising to the use of
$c$-order types instead of types.

\begin{theorem}[Sensing Capacity Theorem for the Contiguous Connections Model] \label{theorem2}
The sensing capacity at distortion $D$ satisfies,
\begin{equation}
C(D) \geq C_{LB}(D) =
\min_{\begin{subarray}{c}\blambda \\
 \lambda_{(0)(1)}+\lambda_{(1)(0)}\geq D \\
\end{subarray}}\frac{D\left(P_{X_iY}^{\bgamma}\|Q_{{X_{i}Y}}^{\blambda}\right)}{H(\tblambda|\blambda')-H(\tbgamma|\bgamma')}
\end{equation}
where $\blambda \in \cp(\{0,1\}^{c},\{0,1\}^{c})$,
$\bgamma_{\ba}=\sum_{\bb\in\{0,1\}^{c}}\lambda_{(\ba)(\bb)}$, and
$H(\tbgamma|\bgamma')=1$.
\end{theorem}

If we specialize this result to the case of $c=1$, this theorem provides a bound that
coincides with our bound for the arbitrary connections model. The proof of the sensing
capacity lower bound is similar for the arbitrary and contiguous connections models. The
main differences in the proofs arise due to the contiguity of sensor field of view,
which necessitates the use of $c$-order types. %The proof is shown in Appendix \ref{AppendixIPSNProof}.
Extensions demonstrated in Section \ref{sec.arbit.extensions} for the arbitrary
connections model can be easily applied to the contiguous connections model.

\emph{Proof Outline: } The proof of Theorem \ref{theorem2} is essentially identical to
the proof of Theorem \ref{theorem1}, with types and joint types replaced by $c$-order
types and joint types. The use of these higher order types requires counting arguments
described in Lemma \ref{lemma:countingtypes}. For $c$-order types, we bound $\alpha$ in
equation (\ref{Pe6}) as follows,
\begin{equation}
\alpha(\bgamma,k)=|\ct_{\bgamma}^{k}|\leq 2^{k H(\tbgamma|\bgamma')}
\end{equation}
For c-order joint types, we bound $\beta(\blambda,k)=|\ct_{\blambda|\bgamma}^{k}|$ in
equation (\ref{Pe6}) using Lemma \ref{lemma:countingtypes}. The set $S_{\bgamma}(D)$ is
defined as,
\begin{equation} \label{eqn.defSDcontig}
S_{\bgamma}(D) = \Bigg\{\blambda: \lambda_{0}+\lambda_{1} \geq D,\
\bgamma_{\ba}=\sum_{\bb\in\{0,1\}^{c}}\lambda_{(\ba)(\bb)}\Bigg\}
\end{equation}
Given these new bounds and definitions, and the substitution of c-order types for types,
the proof of Theorem \ref{theorem1} can be applied directly to prove Theorem
\ref{theorem2}.

\subsection{Numerical Results} \label{sec:ipsn_bounds}

In Figure \ref{indvsdep}, we compare $C_{LB}(D=0.025)$ for sensor networks with localized
(i.e. contiguous connections model) and non-localized (i.e. arbitrary connections model)
sensing. We assume that the sensing function $\Psi$ is a weighted additive function, with
weights $\{1,0.5,0.25,0.1\}$ for $c=4$ and $\{1,0.5,0.25\}$ for $c=3$. The sensor noise
model used throughout this section assumes that the probability of error decays
exponentially with the error magnitude. In the figures, `Noise = $p$' indicates that for
a sensor, $P(Y \ne X) =p$, with ${\cal Y}={\cal X}$ assumed. Contiguous sensor field of
view causes a significant reduction in sensing capacity. We conjecture that this effect
is similar to the inferior performance of channel codes that have finite memory, such as
convolutional codes, as opposed to LDPC codes. Further, it is interesting to note that
the gap in sensing capacity between sensors of range $c=3$ and $c=4$ is larger for the
arbitrary connections model than the contiguous connections model.

To compute the bound shown in Theorem \ref{theorem2}, we solve a sequence of convex
optimization problems. Rather than computing the bound directly, we find the largest $R$
for which the minimum of
$f(\blambda)=D\left(P_{X_iY}^{\bgamma}\|Q_{{X_{i}Y}}^{\blambda}\right)-R(H(\tblambda|\blambda')-H(\tbgamma|\bgamma'))$
over all valid $\blambda$ is greater than $0$. Minimizing $f(\blambda)$ is a convex
optimization problem since $f(\blambda)$ is convex in $\blambda$ and the set of valid
$\blambda$ is convex. Since $H(\tbgamma|\bgamma')=1$, the convexity of $f(\blambda)$ in
$\blambda$ can be proven using the log-sum inequality and the concavity of entropy.

\subsection{Extension to Two-dimensional Fields} \label{sec.2dfields}

The sensing capacity bounds obtained in this section can be extended from discrete target
vectors to two dimensional `target fields.' This extension requires the introduction of
two dimensional types. Such types are histograms over the set of possible two dimensional
patterns. We first analyzed the sensing capacity for a two-dimensional contiguous
connections model in \cite{RachlinISIT05}.

Figure \ref{2d_model} shows an example of our sensor network model. The state of the
environment is modeled as a $k\times k$ grid with $k^2$ spatial positions. Each discrete
position may contain no target or one target, and therefore the target configuration is
represented by a $k^2$-bit target field $\bF$. The possible target fields are denoted
$\bFi$, $i\in\{1,\ldots,2^{k^2}\}$. Target fields occur with equal probability. The
sensor network has $n$ identical sensors. Sensor $\ell$ located at grid block $F_h$
senses a set of contiguous target positions within a Euclidean distance $c$ of its grid
location (though this approach can be extended to other sensor coverage models). Circular
boundary conditions are assumed. Figure \ref{2d_model} depicts sensors with range $c=1$.
Each sensor outputs a value $x \in {\cal X}$ that is an arbitrary function of the targets
which it senses, $x=\Psi(\{f_{v}: v\in\cs_{c,h}\})$, where $\cs_{c,h}$ is the coverage of
a sensor located at grid block $F_h$ with range $c$. Since the number of targets sensed
by a sensor depends only on the sensor range, we write the number of targets in a
sensor's coverage as $|\cs_c|$. We assume a simple model for randomly generating sensor
networks, where each sensor chooses a region of Euclidean radius $c$ with equal
probability among the set of possible regions of radius $c$. This would occur, for
example, if sensors were randomly dropped on a field. All definitions from the
one-dimensional contiguous model extend directly, with target vectors $\bV$ replaced by
fields $\bF$. The rate is defined as $R=\frac{k^2}{n}$.

For a sensor located randomly in the target field, the probability of a sensor producing
a value depends on the number of target patterns that correspond to that value in the
sensor's range, and thus, can be written as a function of the frequency of patterns in
the field. The two-dimensional type $\bgamma_i$ is a vector that corresponds to the
normalized counts over the set of possible target configurations in the sensor's field of
view in a field $\bF_i$. For a sensor of range $c$, $\bgamma_i$ is a $2^{|\cs_c|}$
dimensional vector, where each entry in the vector $\bgamma_i$ corresponds to the
frequency of occurrence of one of the possible $|\cs_c|$ bit patterns. The set of sensor
types $\bgamma$ of a $k\times k$ field is denoted $\cp_{k}^2(\{0,1\}^{|\cs_c|})$.
$\gamma_{(0)}$ and $\gamma_{(1)}$ are the number of zeros and ones respectively in a
vector of type $\bgamma$. These quantities can be directly computed from $\bgamma$.

Next, we note that for sensor of range $c$ the conditional probability $P_{\bxin\bxjn}$
depends on the two-dimensional joint type $\blambda$ of the $i^{th}$ and $j^{th}$ target
fields $\bFi,\bFj$. For $\ba,\bb\in\{0,1\}^{|\cs_c|}$, $\blambda$ is the matrix of
$\lambda_{(\ba)(\bb)}$, the fraction of positions in $\bFi,\bFj$ where $\bFi$ has a
target pattern $\ba$ while $\bFj$ has a target pattern $\bb$. We denote the set of all
joint sensor types for sensors of range $c$ observing a target field of area $k^2$, as
$\cp_{k}^2(\{0,1\}^{|\cs_c|},\{0,1\}^{|\cs_c|})$. Since the output of each sensor depends
only on the contiguous region of targets which it senses, $P_{\bxin\bxjn}$ depends only
on $\blambda$ (discussed in Section \ref{sec:itw_defs}). $\lambda_{(1)(0)}$ is the number
of grid locations where field $i$ has a target and field $j$ does not, and can be
computed directly from $\blambda$. $\lambda_{(0)(1)}$ is similarly defined and computed.

Using the definitions of two dimensional types in the definitions of
$P_{X_iY}^{\bgamma_i}$ and $Q_{{X_{i}Y}}^{\blambda}$ from the one-dimensional contiguous
connections model, we can prove the following bound for sensing a two-dimensional field.
The sensing capacity at distortion $D$ satisfies,
%\begin{senscap}
\begin{equation}
C(D) \geq C_{LB}(D) =
\min_{\begin{subarray}{c}\blambda \\
\lambda_{(0)(1)}+\lambda_{(1)(0)}\geq D
\end{subarray}}\frac{D\left(P_{X_iY}^{\bgamma_i}\|Q_{{X_{i}Y}}^{\blambda}\right)}{H((\gamma_{j(0)},\gamma_{j(1)}))}
\end{equation}
where $\bgamma_{i},\bgamma_{j}\in\cp^{2}(\{0,1\}^{|\cs_c|})$, $\gamma_{i(0)}=0.5$ and
$\gamma_{i(1)}=0.5$, and $\blambda \in \cp^{2}(\{0,1\}^{|\cs_c|},\{0,1\}^{|\cs_c|})$.

\emph{Proof Outline: } The proof is essentially identical to the proof of Theorem
\ref{theorem1}, with types and joint types replaced by two-dimensional types and joint
types. For two-dimensional types, we bound $\alpha$ as follows,
\begin{equation}
\alpha(\bgamma_i,k)\leq 2^{k^2 H((\gamma_{i(0)},\gamma_{i(1)}))}
\end{equation}
For two-dimensional joint types, we bound $\beta$ as,
\begin{equation}
\beta(\blambda,k)\leq 2^{k^2 H((\gamma_{j(0)},\gamma_{j(1)}))}
\end{equation}
The bounds on $\alpha$ and $\beta$ are loose, and the authors are not aware of tighter
combinatorial bounds for two-dimensional types. The set $S_{\bgamma}(D)$ is defined as in
equation (\ref{eqn.defSDcontig}). Given these new bounds and definitions, and the
substitution of 2D types for types, the proof of Theorem \ref{theorem1} can be applied
directly to prove this result.

\section{Conclusions and Discussion}
\label{sec.concs}

The results presented in this paper provide limits on the accuracy of sensor networks for
large-scale detection applications. These results are obtained by drawing on an analogy
between channel coding and sensor networks. We define the sensing capacity and lower
bound it for several sensor network models. For all rates below the sensing capacity,
detection to within a desired accuracy with arbitrarily small error is achievable. This
threshold behavior contrasts with classical detection problems, where probability of
error goes to zero as the number of sensor measurements go to infinity while the number
of hypotheses remains fixed \cite{Cover91}. The sensing capacity captures complex sensor
tradeoffs. For example, our bounds show that the efficiency of using long range, noisy
sensors or shorter range, less noisy sensors depends on the desired detection accuracy.
Further, our results show that the mutual information is not the correct notion of
information for large-scale detection problems. This has implications for the problem of
sensor selection due to the popularity of `information gain' as a sensor selection
metric.

An important contribution of this paper is its demonstration of a close connection
between sensor networks and communication channels. It is thought-provoking to consider
that one could apply insights from the large body of work available for communication
channels to the sensor network setting. For example, channel coding theory contains a
large number of results that are used to build practical communication systems. Can we
fruitfully apply ideas from coding theory to sensor networks? To demonstrate the
potential benefit of a channel coding perspective, in \cite{RachlinIPSN06,
RachlinMILCOM06} we proposed extending ideas from convolutional coding to sensor
networks. We demonstrated that a version of sequential decoding (a low complexity
decoding heuristic for convolutional codes) can be applied to detection in sensor
networks, as an alternative to the belief propagation algorithm. Our empirical results
indicate that above a certain number of sensor measurements, the sequential decoding
algorithm achieves accurate decoding with bounded computations per bit (target position).
This empirical result suggests the existence of a `computational cut-off rate', similar
to one that exists for channel codes.

Our work on the theory of sensing points to a large set of open problems on large-scale
detection. Obvious directions include strengthening the theory by considering alternative
settings of the problem, tightening the sensing capacity bounds, and proving a converse
to sensing capacity. For example, we presented extensions to the work presented in this
paper by considering the impact of spatial \cite{RachlinISIT05} and temporal
\cite{RachlinALLERTON06} dependence on the sensing capacity. Another direction for future
work is to explore the connection between sensor networks and communication channels,
including the exploitation of existing channel codes to design sensor networks.

\bibliographystyle{IEEEtran}
\bibliography{initial_submission}

\newpage

\begin{figure}[t]
\begin{center}
\includegraphics[height=1.6cm]{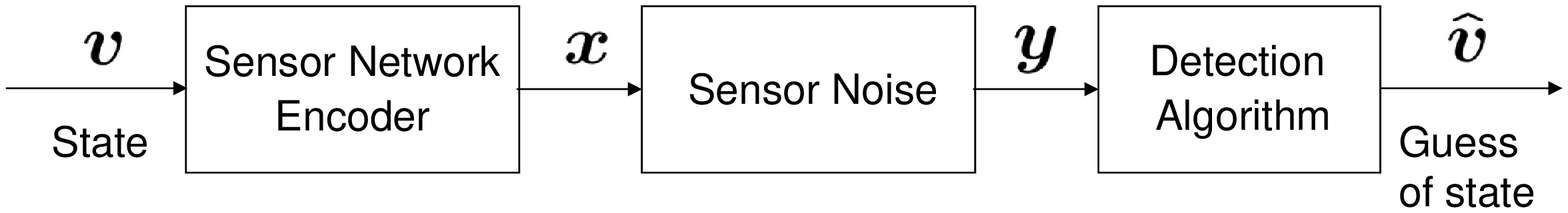}
\caption{Sensor network model.} \label{fig.sensor_channel}
\end{center}
\end{figure}

\begin{figure}[t]
\begin{center}
\includegraphics[height=1.6cm]{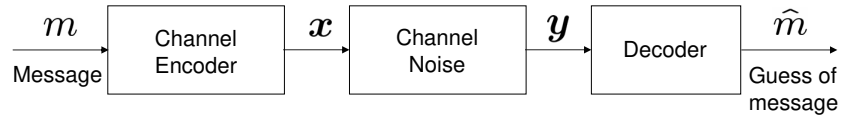}
\caption{Communication channel model.} \label{fig.communication_channel}
\end{center}
\end{figure}

\begin{figure}
\begin{center}
\includegraphics[height=6.0cm]{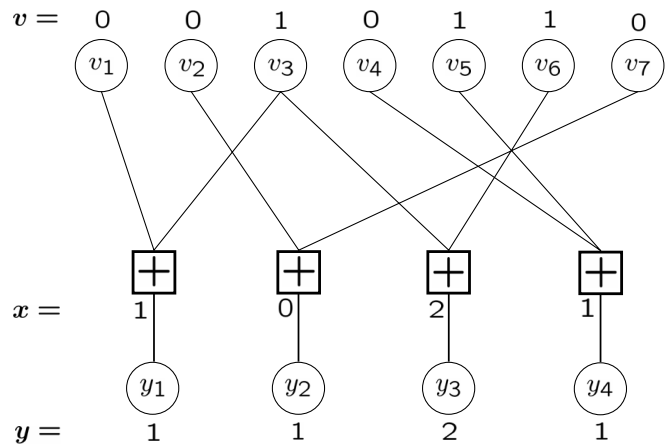}
\caption{Arbitrary connections model with $k=7, n=4, c=2,$ and a sum sensing function.}
\label{fig.ind_model}
\end{center}
%\vspace{-0.25in}
\end{figure}

\begin{table}
\begin{center}
\begin{tabular}{|ccc|}
\hline $\boldsymbol{v_j}$ & $\bgamma$ of $\bVj$  & $\blambda$ of $\boldsymbol{v_j}$ with $\boldsymbol{v_i}= 0010110$ \\
\hline \hline
0010110   &  $\big(\frac{4}{7},\frac{3}{7}\big)$   &  $\big(\frac{4}{7},0,0,\frac{3}{7}\big)$  \\[3pt] %\hline
0000110   &  $\big(\frac{5}{7},\frac{2}{7}\big)$   &  $\big(\frac{4}{7},0,\frac{1}{7},\frac{2}{7}\big)$  \\[3pt] %\hline
1000011   &  $\big(\frac{4}{7},\frac{3}{7}\big)$   &  $\big(\frac{2}{7},\frac{2}{7},\frac{2}{7},\frac{1}{7}\big)$  \\[3pt] %\hline
0000000   &  $\big(1,0\big)$   &  $\big(\frac{5}{7},0,\frac{3}{7},0\big)$  \\%[3pt] %\hline
\hline
\end{tabular}
\end{center}
\caption{ Joint types $\blambda$ for four pairs of target vectors. \label{tab.joint}}{}
%\vspace{-0.2in}
\end{table}

\begin{table}
\begin{center}
\begin{tabular}{|c|ccc|}
\hline $X_i$ & $X_i=0$ & $X_i=1$ & $X_i=2$ \\ \hline
$P_{X_i}$ & $(\gamma_{0})^2$ & $2\gamma_{0}\gamma_{1}$ & $(\gamma_{1})^2$\\
\hline
%
% $X_i$  & $P_{X_i}$
% \\ \midrule %\hline\hline
% $X_i=0$ & $(\gamma_{0})^2$
% \\ %\hline
% $X_i=1$ & $2\gamma_{0}\gamma_{1}$
% \\ %\hline
% $X_i=2$ & $(\gamma_{1})^2$
%\\% \hline
%\botrule
\end{tabular}
\end{center}
\caption{Distribution of $X_i$ in terms of the type  $\bgamma$ of $\bVi$ when
$c=2$.\label{tab.XiTable}}{}
%\vspace{-0.2in}
\end{table}

\begin{table}
\begin{center}
\begin{tabular}{|cccc|}
\hline
 $P_{X_iX_j}$  &            $X_j=0$             & $X_j=1$ & $X_j=2$
 \\ \hline\hline
 $X_i=0$ & $(\lambda_{00})^2$ & $2\lambda_{00}\lambda_{01}$ & $(\lambda_{01})^2$
 \\ %\hline
 $X_i=1$ & $2\lambda_{00}\lambda_{10}$ & $2\left(\lambda_{10}\lambda_{01}+\lambda_{00}\lambda_{11}\right)$  & $2\lambda_{01}\lambda_{11}$
 \\ %\hline
 $X_i=2$ & $(\lambda_{10})^2$            &  $2\lambda_{10}\lambda_{11}$ &  $(\lambda_{11})^2$\\% \hline
\hline
\end{tabular}
\end{center}
\caption{Joint distribution of $X_j$ and $X_i$ in terms of the joint type  $\blambda$ of
$\bVi,\bVj$ when $c=2$.\label{tab.jointTable}}{}
%\vspace{-0.2in}
\end{table}

\begin{figure}[t!]
\begin{center}
\includegraphics[height=8cm,keepaspectratio=true]{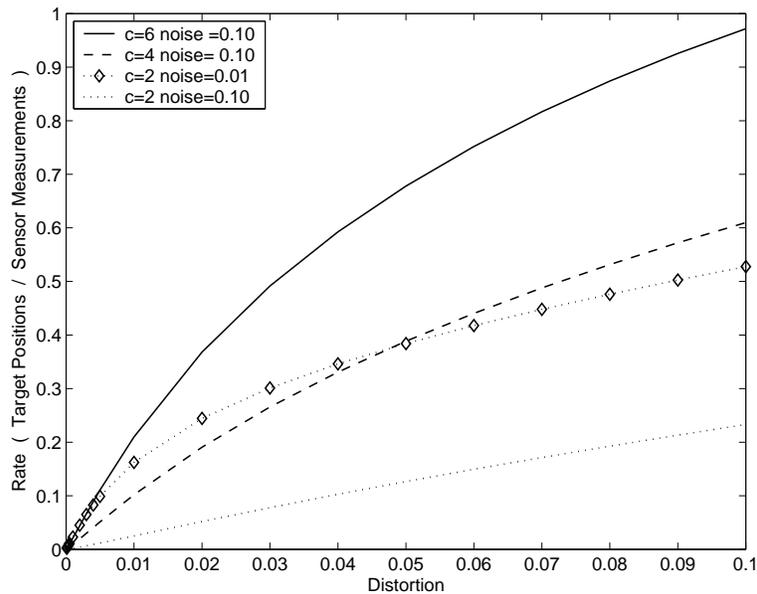}
%\vspace{-0.1in}
\caption{$C_{LB}(D)$ of arbitrary connections model for sensors of varying noise levels
and range.} \label{plot.RD_indplot}
\end{center}
%\vspace{-0.25in}
\end{figure}

\begin{figure}[t!]
\begin{center}
\includegraphics[height=8cm,keepaspectratio=true]{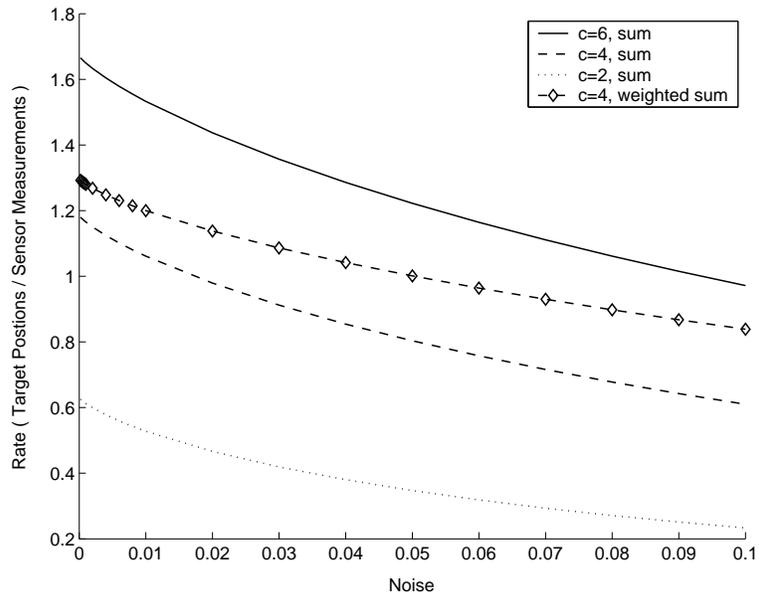}
%\vspace{-0.1in}
\caption{$C_{LB}(0.1)$ of arbitrary connections model for sensors of varying noise
levels, range, and sensing function.} \label{plot.RN_indplot}
\end{center}
%\vspace{-0.25in}
\end{figure}

\begin{figure}[t!]
\begin{center}
\includegraphics[height=8cm,keepaspectratio=true]{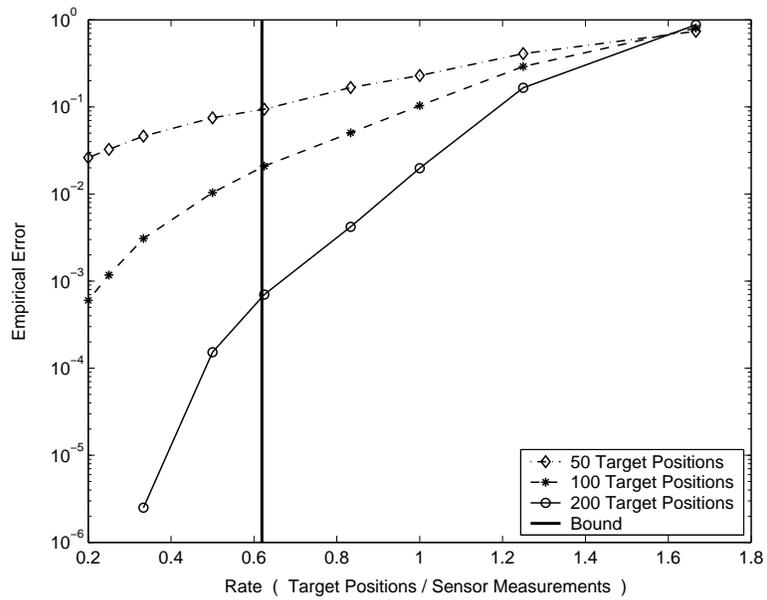}
%\vspace{-0.15in}
\caption{Average empirical error rate of belief propagation based detection for varying
rates, and the corresponding sensing capacity bound.} \label{bp_plot}
\end{center}
%\vspace{-.25in}
\end{figure}

\begin{figure}[t]
\begin{center}
\includegraphics[height=5.2cm]{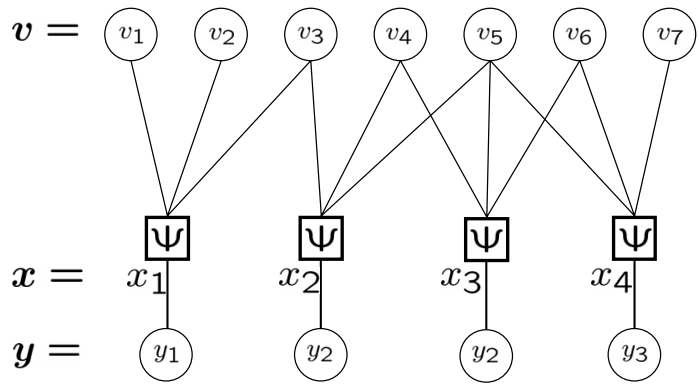}
\caption{Sensor network model with $k=7, n=3, c=3$, contiguous connections, and a sensing
function corresponding to the weighted sum of the observed targets.}
\label{fig:dep_model}
\end{center}
%\vspace{-0.3in}
\end{figure}

\begin{table}[t]
\begin{center}
\begin{tabular}{|ccccc|}\hline
$\lambda_{(ab)(cd)}$ & $cd=00$ & $cd=01$ & $cd=10$ & $cd=11$\\
\hline\hline $ab = 00$ & 0 & 0 & 1/8 & 2/8 \\
 $ab = 01$ & 1/8 & 1/8 & 0 & 0\\
 $ab = 10$ & 1/8 & 1/8 & 0 & 0 \\
 $ab = 11$ & 0 & 0 & 1/8 & 0 \\ \hline
\end{tabular}
\end{center}
%\vspace{-0.2in}
 \caption{ $\blambda$ with $c=2$ for $\bVi =
01101000$ and $\bVj = 01000111$.} \label{tab:lambda}
%\vspace{-0.2in}
\end{table}

\begin{table*}[t]
\begin{center}
\begin{tabular}{|cccc|}\hline
 $P_{X_iX_j}$  &            $X_j=0$             & $X_j=1$ & $X_j=2$
 \\ \hline\hline
 $X_i=0$ & $\lambda_{(00)(00)}$ & $\lambda_{(00)(01)}+\lambda_{(00)(10)}$ & $\lambda_{(00)(11)}$
 \\
 $X_i=1$ & $\lambda_{(10)(00)}+\lambda_{(01)(00)}$ & $\lambda_{(01)(01)}+\lambda_{(01)(10)}+\lambda_{(10)(01)}+\lambda_{(10)(10)}$  & $\lambda_{(10)(11)}+\lambda_{(01)(11)}$
 \\
 $X_i=2$ & $\lambda_{(11)(00)}$    &  $\lambda_{(11)(01)}+\lambda_{(11)(10)}$   &  $\lambda_{(11)(11)}$    \\ \hline
\end{tabular}
\end{center}
%\vspace{-0.2in}
\caption{Joint distribution of $X_j$ and $X_i$ in terms of the joint type $\blambda$ of
$\bVj$ and $\bVi$, with $c=2$. }
\label{tab:XiXj} %\vspace{-0.2in}
\end{table*}

\begin{figure}[t!]
\begin{center}
\includegraphics[height=8cm,keepaspectratio=true]{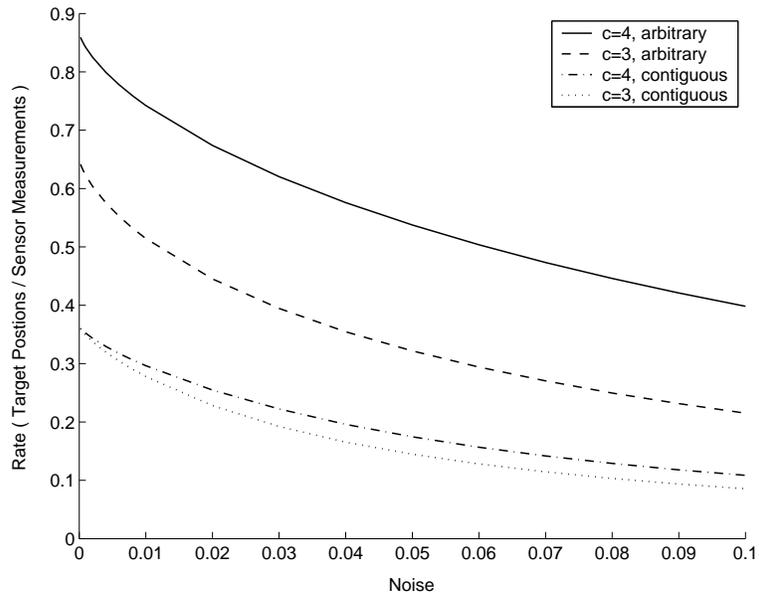}
%\vspace{-0.15in}
\caption{$C_{LB}(0.025)$ for localized and non-localized sensors.} \label{indvsdep}
\end{center}
%\vspace{-0.2in}
\end{figure}

\begin{figure}[t]
\begin{center}
\includegraphics[height=6.5cm]{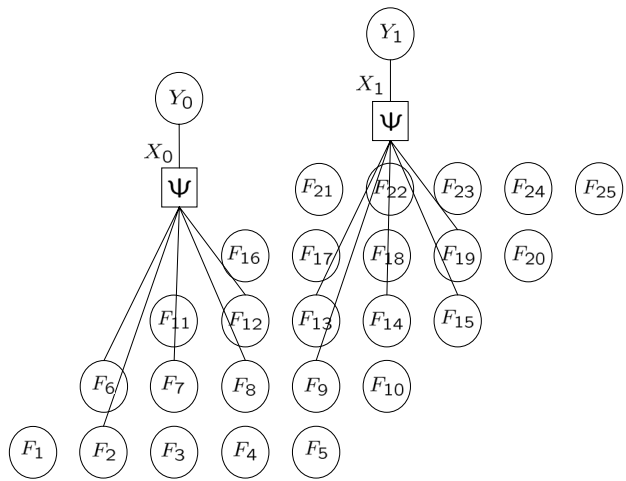}
\caption{Sensor network model with $k=5, n=2, c=1$.}
%The
%environment is shown as a graph where the dependencies among
%targets are shown explicitly. Incomplete arcs at the edge of the
%grid correspond to the circular boundary conditions. }
\label{2d_model}
\end{center}
\vspace{-0.3in}
\end{figure}

\end{document}